\documentclass[a4paper,11pt]{article}
\usepackage{amsmath,amsfonts,amssymb}
\usepackage[utf8]{inputenc}
\usepackage{multirow}
\usepackage{listings}
\usepackage[T1]{fontenc}
\usepackage{times}
\usepackage{multicol,amssymb,amsthm,amsmath,graphicx,geometry,xcolor}
\usepackage{hyperref}
\usepackage{float}
\usepackage{mathrsfs} 
\newtheorem{theorem}{Theorem}

\newtheorem{proposition}{Proposition}
\newtheorem{definition}{Definition}

\newtheorem{lemma}{Lemma}
\newtheorem{remark}{Remark}
\newtheorem{corollary}{Corollary}

\providecommand{\keywords}[1]
{
  \small	
  \textbf{\textit{Keywords:}} #1
}
\providecommand{\Classification}[1]
{
  \small	
  \textbf{\textit{2010 Mathematics Subject Classification:}} #1
}

\numberwithin{equation}{section}
\title{American Passport options in an exponential L\'evy model}
\author{
 Marah Zakaria \\
  \texttt{marahzakaria1@gmail.com} \\
}
\begin{document}
\maketitle
\begin{abstract}
In this paper we examine the problem of valuing an exotic
derivative known as the American passport option where the underlying is driven by a L\'evy process. The passport option is a call option on a trading account. We derive the pricing equation, using the dynamic programming principle,
and prove that the option value is a viscosity solution of variational inequality. We also establish the comparison principle, which yields
uniqueness and the convexity of the viscosity solution.
\end{abstract}
\keywords{L\'evy process, American Passport options, Viscosity solution}.\\
\Classification{35D40,91B25}.
\section{Introduction}
Hyer et al. \cite{Hyer} introduced passport option, that allows the holder to adopt a trading strategy involving securities, while being compensated for the losses resulting from the adoption of such a strategy. They construct the Hamilton–
Jacobi–Bellman (HJB, in short) equation and solve the PDE via Green’s functions. They also define a non-symmetric version
of the problem for which the PDE is solved numerically. Andersen et al. \cite{Andersen} use the
change of num\'eraire to give the one-dimension PDE for the European type and extend it to
the American type. Bian et al. \cite{Bian} establish the mathematical
foundation for pricing the European passport option. They prove the comparison principle, uniqueness and convexity preserving of the viscosity solutions of related HJB. Wang et al. \cite{Wang} study the pricing problem for the European passport options
in a jump-diffusion model by using the method of viscosity solutions. In another paper  Wang et al. \cite{Wang2} study the pricing problem for the American passport options. This paper is concerned with the American passport option. In the work Wang et al. \cite{Wang2}, the underlying stock follows geometric Brownian
motion. However, various empirical studies show that such models seem to dissatisfy
the real market behavior. The purpose of this paper is to give a pricing analysis for the American passport option when the underlying follows a L\'evy process which include the case of Jump diffusion model. The purpose of this paper is to give a pricing analysis for
the European passport option in a jump-diffusion model by a viscosity approach. To
our knowledge, the American option with jump-diffusion processes has not appeared
in the literature before.\\
This paper is organized as follows: In Sect. 2, we rigorously establish the mathematical model for the American passport option pricing problem. We show that the price of
the American passport option verifies a variational inequality. In Sect. 3,
we prove the the existence, the uniqueness and the convexity of the viscosity solution.  
\section{Model formulation}
A L\'evy process $X$ is a c\'adl\'ag  stochastic process with values in $\mathbb{R}$ ,
starting from 0, with stationary and independent increments. The random process $X$
can be interpreted as the independent superposition of a Brownian motion with drift
and an infinite superposition of independent (compensated) Poisson processes. More
precisely we give the following L\'evy–Itô decomposition presentation
of X:
\begin{align}\label{eq:X}
dX_{t}&=\big(r-a-\frac{\sigma^2}{2}+\int_{\mathbb{R}} (e^{z}-1)-z1_{z<1}\nu(dz)\big)dt+\sigma dW_{t}\\&+\int_{|z|<1} z (J(dt,dz)-\nu(dz)dt\big)+\int_{|z|>1} z J(dt,dz)
\end{align}
In Equation \eqref{eq:X}, $W_{t}$ is a Brownian motion process which is adapted to the filtration. The component $J(ds,dz)$ is a Poisson random measure with intensity measure  $\nu(dz)$. The measure $\nu$ is a positive Radon measure on $\mathbb{R} \setminus {0}$, called the Lévy measure, and it satisfies
\begin{align*}
\int_{|z|>1} e^{z} \nu (dz)<\infty,
\end{align*}
$J(ds,dz)-\nu (ds,dz)$ is the compensated Poisson random measure that corresponds to $J(ds,dz)$. \\
We will from now on assume $\mathbb{P}$ to be a risk-neutral measure, the interest
rate to be a constant r and a constant
positive dividend yield $a$, per annum.\\
We suppose the risky asset ${(S_t)}_{0\leq t\leq T}$ and defined on a filtered probability space $\{\Omega,\mathcal{F}, \mathcal{F}_t,\mathbb{P}\}_{0\leq t\leq T}$, evolves according to the following expression:
\begin{align}\label{eq:S}
S_{t}=e^{X_t}
\end{align}
Applying Ito lemma:
\begin{align}\label{eq:S}
\frac{dS_{t}}{S_{t}}=(r-a)dt+\sigma dW_{t}+\int_{\mathbb{R}} (e^{z}-1)(J(dt,dz)-\nu(dz)dt)
\end{align}
The passport option, introduced and marketed by Bankers Trust, is a call option
on the balance of a trading account. Denote the option holder’s trading strategy by $q_t$ , the
number of shares held at time $t$. $q_t$ can be positive (if the option holder thinks the
asset price will rise) or negative (if he/she thinks the asset price will fall). There is a
constraint on $q_t$ : $|q_t| \leq C$.  At expiry T , the issuer
has to pay $X^+_T = max(X_T , 0)$ to the option holder, where $X_t$ represents the value of the trading account.
\begin{align*}
dX_{t}&=q_t(dS_t-rS_tdt)+r X_tdt
\end{align*}
 We are interested in American-type of Passport derivative of the form:
\begin{align}\label{problem}
V(t,S_{t}, X_{t})=\sup_{|q|\leq C}\sup_{\tau \in \mathcal{T}_{t,T}} \mathbb{E}[e^{-r(\tau -t )}X_T^+|\mathcal{F}_t]
\end{align}
Throughout this paper, we assume that $C = 1$. We do not lose generality because $V_{C=C^*}(t,S_{t}, X_{t}) = V_{C=1}(t,C^*S_{t}, X_{t})$
\begin{proposition}
Let 
\begin{align*}
V^q(t,S_{t}, X_{t})=\sup_{\tau \in \mathcal{T}_{t,T}} \mathbb{E}[e^{-r(\tau -t )}X_T^+|\mathcal{F}_t]
\end{align*}
We denote $L_t=\frac{X_t}{S_t}$ Then we have the following relationship:
\begin{align*}
V^q(t,S_{t}, X_{t})=S_t\sup_{\tau \in \mathcal{T}_{t,T}} \mathbb{E}^{\mathbb{Q}}[e^{-a(\tau-t )}L_T^+|\mathcal{F}_t]
\end{align*}
where $\mathbb{Q}$ is defined by:
\begin{align*}
\frac{d\mathbb{Q}}{d\mathbb{P}}=\frac{S_{T}}{S_{0} e^{-(r-a)T}}
\end{align*}
\end{proposition}
\noindent As a conclusion we can reduce the problem \eqref{problem} by one dimension:
\begin{align*}
V(t,S_{t}, X_{t})=S_tu(t,L_t)
\end{align*}
where
\begin{align}\label{problem}
u(t,L_t)=\sup_{|q|\leq 1}\sup_{\tau \in \mathcal{T}_{t,T}} \mathbb{E}[e^{-a(\tau -t )}L_T^+|\mathcal{F}_t]
\end{align}
\begin{proposition}
Let $f(t,x)= \sup_{|q|\leq 1}\mathbb{E}[e^{-r(\tau -t )}X_T^+|\mathcal{F}_t]$ then $f$ verifies the following IPDE
\begin{align*}
0=&\partial_t f(t,L_t) dt + \sup_{|q|\leq 1}\Big \{\sigma^2\frac{(q_t-L_t)^2}{2}\partial_{ll} f(t,L_t)-(q_t-L_t)\partial_x f(t,L_t)\\
&+\int_{\mathbb{R}}f(t,L_{t-}+(q_t-L_t)\frac{e^z-1}{e^z},)-f(t,X_{t-})-(q_t-L_t)\frac{e^z-1}{e^z}\partial_L f(t,L_{t-})\tilde{\nu}(dt,dz)\Big \}-af(t,L_t) 
\end{align*} 
\end{proposition}
\begin{proof}
\noindent Apply Ito’s formula to get
\begin{align*}
d\frac{X_t}{S_t}=\big(q_t-\frac{X_t}{S_t}\big)\big(\sigma dW_t-(\sigma+a) dt+\int_{\mathbb{R}}\frac{e^z-1}{e^z}J(dt,dz)-e^z\nu(dz)dt\big)
\end{align*}
We define the following change (cf. Jacod and Shiryaev \cite{Jacod})
\begin{align*}
\frac{d\mathbb{Q}}{d\mathbb{P}}=\frac{S_{T}}{S_{0} e^{-(r-a)T}}
&=\mathcal{E}\Big(\int_0^T \sigma dW_{s}+\int_0^T \int_{\mathbb{R}} (e^{z}-1)(J(ds,dz)-\nu (dz)ds)\Big)
\end{align*}
\begin{align*}
dW^{\mathbb{Q}}_{t}&=- \sigma dt+ dW_{t},\\
\tilde{\nu}&=e^{z} \nu,
\end{align*}
then
\begin{align*}
dL_t=\big(q_t-L_t\big)\big(-adt+\sigma dW_t^{\mathbb{Q}}+\int_{\mathbb{R}}\frac{e^z-1}{e^z}\big(J(dt,dz)-\tilde{\nu}(dz)dt\big).
\end{align*}
Let $e^{-at}f(t,L_t)=\mathbb{E}^{\mathbb{Q}}[e^{-aT}L_T^+|\mathcal{F}_t]$, which is a martingale by definition. Using Ito lemma,
\begin{align*}
de^{-at}f(t,L_t)&=e^{-at} \Big [\partial_t f(t,L_t) dt+\partial_l f(t,L_t) dL_t 
+ \sigma^2\frac{(q_t-L_t)^2}{2}\partial_{ll} f(t,L_t) -af(t,L_t)dt
\\&+\int_{\mathbb{R}}u(t,L_{t-}+(q_t-L_t)\frac{e^z-1}{e^z})-f(t,L_{t-}) \tilde{J}(dt,dy)\\
&+\int_{\mathbb{R}}f(t,L_{t-}+(q_t-L_t)\frac{e^z-1}{e^z})-u(t,L_{t-})-(q_t-L_t)\frac{e^z-1}{e^z}\partial_L u(t,L_{t-})\tilde{\nu}(dt,dz)\Big].
\end{align*} 
Using the martingal property then we can conclude our result.
\end{proof}
\noindent The Hamilton-Jacobi-Bellman (HJB in short) equation associated
with the problem \eqref{problem} is a variational inequality involving, at least heuristically, a non-linear second order parabolic integro-differential equation (see
Bensoussan, J.L. Lions \cite{Bensoussan}
\begin{align}\label{HJB}
\begin{cases}
min\{-\partial_t u(t,x)-\mathcal{L}u(t,x)+au(t,x),u(t,x)-x\}=0 \text{ }\forall (t,x)\in [0,T)\times \mathbb{R}\\
u(T,x)=x^+
\end{cases}
\end{align}
where
\begin{align*}
\mathcal{L}u(t,x)&=\sup_{|q|\leq 1}\Big \{\sigma^2\frac{(q_t-L_t)^2}{2}\partial_{ll} u(t,L_t)-(q_t-L_t)\partial_x u(t,L_t)\\
&+\int_{\mathbb{R}}u(t,L_{t-}+(q_t-L_t)\frac{e^z-1}{e^z})-u(t,X_{t-})-(q_t-L_t)\frac{e^z-1}{e^z}\partial_L u(t,L_{t-})\tilde{\nu}(dt,dz)\Big \}
\end{align*} 
\begin{remark}
If the risky asset ${(S_t)}_{0\leq t\leq T}$ is under jump-diffusion model:
\begin{align}\label{eq:S}
\frac{dS_{t}}{S_{t}}=(r-a)dt+\sigma dW_{t}+\int_{\mathbb{R}} e^{z}-1 J(dt,dz),
\end{align}
we write the operator  $\mathcal{L}$ as following:
\begin{align*}
\mathcal{L}u(t,x)&=\sup_{|q|\leq 1}\Big \{\sigma^2\frac{(q_t-L_t)^2}{2}\partial_{ll} u(t,L_t)-(q_t-L_t)\partial_x u(t,L_t)\\
&+\int_{\mathbb{R}}u(t,L_{t-}+(q_t-L_t)\frac{e^z-1}{e^z})-u(t,X_{t-}) \tilde{\nu}(dt,dz)\Big \}
\end{align*} 

\end{remark}
\noindent It is well known that the value function $u$ defined in \eqref{HJB} is not always smooth
enough for the expression $u$ to be well defined. Therefore, it is natural to ask for a weak solution such that $u$ is unique even though it is not smooth. One such weak solution, called a viscosity solution, was introduced by Crandall and Lions see
Crandall et al. \cite{Grandall}.\\
We shall need the following spaces of semi-continuous functions on $[0,T]\times \mathbb{R}$ :
\begin{align*}
USC([0,T]\times \mathbb{R}) :=\Bigg \{ v:[0,T]\times \mathbb{R}\rightarrow \mathbb{R}\cup \{ -\infty\}, \text{it is upper semi-continuous} \Bigg \}\\
LSC([0,T]\times \mathbb{R}) :=\Bigg \{ v:[0,T]\times \mathbb{R}\rightarrow \mathbb{R}\cup \{ +\infty\}, \text{it is lower semi-continuous} \Bigg \}
\end{align*} 
\begin{definition}
 Let $u$ be a locally bounded function:
\begin{itemize}
\item A viscosity sub-solution (super-solution) of \eqref{HJB} if  for any $\psi \in \mathcal{C}^{1,2}([0, T] \times \mathbb{R})$, wherever $(t,x) \in [0, T] \times \mathbb{R}$ is a maximum (minimum) of $u-\psi$
\begin{align*}
max \Big(\partial_t \psi(t,x)+\mathcal{L}\psi(t,x) -a\psi(t,x),x -\psi(t,x)\Big)\leq 0(\geq)
\end{align*}
\item  u is a viscosity solution of \eqref{HJB}if it is both super and subsolution.  
\end{itemize}
\end{definition}
\section{Viscosity solutions}
\subsection{Continuity}
This section is devoted to the conitnuity property of the value function $u$. 
\begin{proposition}
u is continuous and
\begin{align}\label{linear}
|u(t,x)-u(s,y)|\leq C(\sqrt{t-s}+|x-y|)
\end{align}
\end{proposition}
\begin{proof}
We have
\begin{align*}
\mathbb{E}^\mathbb{Q}\Big[|L^{x}_t-L^{y}_t|^2\Big]&\leq C\Big( |x-y|^2+\mathbb{E}^\mathbb{Q}\Big[\int_0^t q_u|L^{x}_s-L^{y}_s|^2ds\Big]+\mathbb{E}^\mathbb{Q}\Big[\int_0^t\int_{\mathbb{R}} |L^{x}_s-L^{y}_s|^2 |e^z-1|^2\tilde{\nu}(dz)ds\Big]\Big)\\
&\leq C\big( |x-y|^2+\mathbb{E}^\mathbb{Q}\Big[\int_0^t |L^{x}_s-L^{y}_s|^2ds\Big]\big)
\end{align*}
Using Gronwall lemma we conclude that
We have
\begin{align*}
\mathbb{E}^\mathbb{Q}\Big[|L^{x}_t-L^{x}_t|^2\Big]\leq C( |x-y|^2)
\end{align*}
then
\begin{align*}
|u(t,x)-u(t,y)|&= \sup_{|q|\leq 1}\sup_{\tau \in \mathcal{T}_{t,T}} \mathbb{E}[e^{-a(\tau -t )}\big|(L_T^{t,x})^+-(L_T^{t,y})^+\big||\mathcal{F}_t]\\
&\leq C \mathbb{E}^\mathbb{Q}\Big[|L^{x}_t-L^{y}_t|\Big]\\
&\leq C( |x-y|)
\end{align*}
We now show continuity with respect to time for fixed $x$. Let $0 \leq t \leq t' \leq T$. Take $\tau \in \mathcal{T}_{T-t}$ and define $\tau'=\tau \wedge (T-t')$  . We note that $\tau' \in \mathcal{T}_{T-t'}$ and $\tau' \leq \tau\leq \tau' +t'-t$
\begin{align*}
|u(t,x)-u(t',x)|&\leq C  \sup_{\tau' \leq s\leq \tau' +t'-t} \mathbb{E}^\mathbb{Q}[|L^{t,x}_s-L^{t,x}_{\tau'}|]\\
&\leq C\sqrt{(1+x^2)|t'-t|}
\end{align*}
Where we used the fact that
\begin{footnotesize}
\begin{align}\label{eq:cont}
 \sup_{\tau \leq s\leq \tau' }\mathbb{E}^\mathbb{Q}[|L_s^{t,x}-L_{\tau}^{t,x}|^2]\leq C(\tau'-\tau)(1+|x|^2)
\end{align}
\end{footnotesize}
So $u$ is continuous with respect to time.
\end{proof}
The dynamic programming principle for the optimal stopping of a controlled process is the following (Proposition 3.2 in \cite{Pham}):
\begin{lemma}\label{lemma3}
Let $\delta>0$,  we define the stopping time as 
\begin{align*}
\tau^{\delta}_{q,t,x} =\inf \{0\leq s \leq T-t | u(t+s,X_s^{x})\leq X_s^{x}+\delta\}
\end{align*}
For any $\tau \leq \tau^{\delta}_{q,t,x}$
\begin{align*}
u(t,x) =\sup_{|q|\leq 1}\mathbb{E} [e^{-a\tau} u(t+\tau,X_{\tau}^{x})]
\end{align*}
\end{lemma}
\noindent Now we are ready to give the existence result of a viscosity solution.
\begin{theorem}
u is a viscosity solution of IPDE \eqref{HJB}.
\end{theorem}
\begin{proof}
We already know that $u$ is continuous.
Let’s start by showing that $u$ is a super-solution of \eqref{HJB}. Let $(t, x) \in [0, T )\times  \mathbb{R}$ and $\psi \in \mathcal{C}^{1,2}(\mathcal{O})$ such that without loss of generality
\begin{align}\label{eq:super}
0=(u-\psi)(t,x,v)=\min_{[0, T )\times \mathbb{R}} u-\psi.
\end{align}
Applying lemma \eqref{lemma3} , for any stopping time $\epsilon \in [t,T]$
\begin{align*}
u(t,x)\geq \mathbb{E} [e^{-a\epsilon} u(t+\epsilon,X_{\epsilon}^{x})].
\end{align*}
By \eqref{eq:super}
\begin{align*}
0 \geq \mathbb{E} [e^{-a\epsilon} \psi(t+\epsilon,X_{\epsilon}^{x})-\psi(t,x)]
\end{align*}
Using Ito lemma we have
\begin{align*}
0\geq  \mathbb{E} [\int_0^{\epsilon} \Big(\partial_t  +\mathcal{L}-a\Big)\psi(t+s,X_{s}^{x}) ds],
\end{align*}
dividing by  $\epsilon$ and taking the limit as $\epsilon \rightarrow 0$,
\begin{align*}
0\geq  (\partial_t  +\mathcal{L}-a)\psi(t,x)
\end{align*}
We next prove that $u$ is a subsolution of \eqref{HJB}. Let $(t, x) \in [0, T )\times \mathbb{R}$ and $\psi \in \mathcal{C}^{1,2}([0,T[\times \mathbb{R})$ such that without loss of generality
\begin{align}\label{eq:sub1}
0=(u-\psi)(t,x)=\max_{\bar{\theta}\in [0, T )\times \mathbb{R}_+} (u-\psi)(\bar{\theta})
\end{align}
We already know that $u(t,x) \geq x$. If $u(t,x) = x$, the inequality of
subsolution is obviously satisfied. Assume therefore that $u(t,x) > x$.  and define
\begin{align*}
\delta=\frac{u(t,x)-x}{2}>0
\end{align*}
For each admissible control $|q|\leq 1$, we define the stopping time
\begin{align*}
\tau^{\delta_m} =\inf \{0\leq s \leq T-t | u(t+s,X_s^{x} )\leq X_s^{t,x}+\delta_m)\}
\end{align*}
such that $\delta_m \to 0$. To prove the subsolution property, we assume on the contrary that 
\begin{align*}
(-\partial_t-\mathcal{L}+a\big)\psi(t,x)>0 \hspace{0.5cm} \text{ and} \hspace{0.5cm} u(t,x)-x>0 
\end{align*}
As the operator $\partial_t+\mathcal{L}$ is continuous at $(t,x)$ then there exists an $\eta>0$ and an $\epsilon>0$ such that
\begin{align*}
(-\partial_t-\mathcal{L}+a\big)\psi(s,y)\geq \epsilon   
\end{align*}
for all $s,y$ in $B((t,x),\eta)= \{ (s,y)\in  [0, T )\times \mathbb{R}_+ :| y-x|+|s-t| \leq \eta \}$. Let 
\begin{align*}
\tau_0 =\inf \{0\leq s \leq T-t | | X_s^{x} -x|+s>\eta \}
\end{align*}
We define $\tau=\tau_0\wedge \tau^{\delta_m}$
\begin{align}\label{eq:sub2}
u(t,x)= \mathbb{E} [e^{-a\tau} u(t+\tau,X_{\tau}^{t,x})]
\end{align}
Using \eqref{eq:sub1}, \eqref{eq:sub2} and the definition of $\tau_0$:
\begin{align}\label{eq:sub3}
0&\leq \mathbb{E} [e^{-a\tau} \psi(t+\tau,X_{\tau}^{t,x})-\psi(t,x)]\nonumber \\
&=\mathbb{E}[\int_0^{ \tau} (\partial_s+\mathcal{L}-a) \psi(t+s,X_{s}^{t,x})   ds
\leq -\epsilon \mathbb{E}[\tau].
\end{align}
Using Tchebyshev's inequality we deduce 
\begin{align*}
\mathbb{Q}(\tau_0 \leq \tau^{\delta_m} )&\leq \mathbb{Q}(\sup_{0\leq s \leq \tau^{\delta_m}} |X^{x}_s-x|>\eta)\\
&\leq  \frac{1}{\eta^2}\mathbb{E}^{\mathbb{Q}}[ \sup_{0\leq s\leq \tau^{\delta_m}}|X^{x}_s-x|^2]\xrightarrow[m \to \infty]{} 0.
\end{align*}
Moreover, since
\begin{align*}
\mathbb{Q}(\tau_0 > \tau^{\delta_m} )\leq  \frac{1}{\tau^{\delta_m}}\mathbb{E}^{\mathbb{Q}} [\tau_0\wedge \tau^{\delta_m}]\leq 1,
\end{align*}
this implies that $\frac{1}{\tau^{\delta_m} }\mathbb{E}[\tau]$ converges to 1 when $m \to \infty$. We conclude from \eqref{eq:sub3} that $\epsilon \leq 0$ which is a contradiction. Then $\psi$ is a subsolution. 
\end{proof}
\subsection{Uniqueness and convexity}
In this section, we prove a comparison result from which we obtain the uniqueness of the
solution of the IPDE \eqref{HJB}. We discuss the convexity of viscosity solutions using the comparison technique. In proving the uniqueness result and convexity for viscosity solutions of second order
equations, we will need a fair amount of notation.  We define the second order sub and super "jets" $J_{}^{2,+}$, $J_{}^{2,-}$
of mappings $u\in \mathcal{C}([0,T]\times \Omega)$. We set 
\begin{align*}
J_{}^{2,+}=&\Big\{ (p_0,p,X)\in [0,T] \times\mathbb{R}\times \mathbb{R}: u(y)\leq u(x)+p_0(t-s)+p(y-x)\\
&+\frac{1}{2}X(y-x)^2+o(|y-x|^2) \text{ as } (s,y)\rightarrow (t,x) \Big\}
\end{align*}
and
\begin{align*}
J_{}^{2,-}=&\Big\{ (p_0,p,X)\in [0,T] \times\mathbb{R}\times  \mathbb{R}: u(y)\geq u(x)+p_0(t-s)+p(y-x)\\
&+\frac{1}{2}X(y-x)^2+o(|y-x|^2) \text{ as } (s,y)\rightarrow (t,x)\Big\}
\end{align*}
\begin{theorem}\label{theorem2}
Let $u$ (resp. $v$) $\in USC\big([0; T ] \times\mathbb{R}\big)$ (resp. $\in LSC\big([0; T ] \times\mathbb{R}\big)$) be a viscosity subsolution (resp. supersolution) of \eqref{HJB}. If $u(T,x) \leq v(T,x)$
for all $x\in \mathbb{R}$
, then 
\begin{align*}
u(t, x) \leq v(t, x) \hspace{0.3cm}\forall [0; T ] \times \mathbb{R}
\end{align*}

\end{theorem}
\begin{proof}
let us define the function  $\psi \in \mathcal{C}^{1,2}([0, T]\times \mathbb{R}) $. We introduce the function
\begin{align}\label{eq:three}
\Psi(t,x,y)=u(t,x)-v(t,y)-\frac{|x-y|^2}{2\alpha}-\frac{\theta}{2}e^{\lambda(T-t)}(x^2+y^2)
\end{align}
where $\epsilon$, $\alpha$ are positive parameters which are devoted to tend to zero.\\
Since $u, v \in \mathcal{C}([0, T ] \times \mathbb{R}$, $\Psi$ admits a maximum at  $(\tilde{t}, \tilde{x},\tilde{y})$ such that is a global maximum point of $\Psi$. We have that
$\Psi(\tilde{t}, \tilde{x}, \tilde{x}) +\Psi(\tilde{t}, \tilde{y}, \tilde{y})\leq 2\Psi(\tilde{t} , \tilde{x},\tilde{y})$ and using \eqref{linear}, which give
\begin{align}\label{eq1}
\frac{|\tilde{x}-\tilde{y}|^2}{2\alpha}&\leq 2C(1+\tilde{x}+\tilde{y})
\end{align}
$\Psi(T, 0, 0) \leq \Psi(\tilde{t} , \tilde{x},\tilde{y})$and  \eqref{linear} give
\begin{align}
\frac{\theta}{2}(\tilde{x}^2+\tilde{y}^2)&\leq 
v(T,0)-u(T,0)+u(\tilde{t},\tilde{x})-v(\tilde{t},\tilde{y})\nonumber\\
&\leq 2C(1+\tilde{x}+\tilde{y})
\end{align}
We deduce that there exists a constant C , depending on $\theta$, such that:
\begin{align}\label{eq2}
\tilde{x},\tilde{y}\leq C_{\theta}
\end{align}
It follows from \eqref{eq1}and \eqref{eq2} that, along a subsequence  $(\tilde{t},\tilde{x},\tilde{y})$ converges to $(t_0, x_0, y_0) \in [0, T ] \times \mathbb{R}$ as $\alpha$ go to zero and $x_0=y_0$.\\
If $\tilde{t}= T$ then knowing that $\Psi(t, x, x) \leq \Psi(T , \tilde{x},\tilde{y})$, we have
\begin{align*}
u(t,x)-v(t,x)-\theta e^{\lambda(T-t)}x^2\leq u(T,\tilde{x})-v(T,\tilde{y})
\end{align*}
sending $\alpha$ and $\theta$ to zero give us $u(t,x)\leq v(t,x)$.\\
In the remaining of the proof we consider the case $\tilde{t}< T$. Applying Theorem 8.3 of Crandall-Ishii  \cite{Grandall} to the function $\Psi$ at point $ (\tilde{t},\tilde{x},\tilde{y})$, we can find real numbers $p_0$, $\tilde{X}$ and $\tilde{Y} \in \mathbb{R}$
such that
\begin{align*}
\Big(p_0-\lambda \frac{\theta}{2} e^{\lambda(T-t)}(\tilde{x}^2+\tilde{y}^2),p_x,\tilde{X}\Big)\in J^{2,+}u(\tilde{t},\tilde{x})&\\
\Big(p_0,p_y,\tilde{Y}\Big)\in J^{2,-}v(\tilde{t},\tilde{y})
\end{align*}
with 
\begin{align*}
p_x= \frac{|\tilde{x}-\tilde{y}|}{\alpha}+\theta e^{\lambda(T-\tilde{t})}\tilde{x} \hspace{1cm} p_y=\frac{|\tilde{x}-\tilde{y}|}{\alpha}-\theta e^{\lambda(T-\tilde{t})}\tilde{y}
\end{align*}
Since $u$ and $v$ are respectively sub and supersolution of \eqref{HJB}, we have for some $\phi \in \mathcal{C}^2 ([0,T]\times \mathbb{R})$
\begin{align}\label{eq:one}
\min\Big(&au(\tilde{t},\tilde{x})-p_0+\lambda \frac{\theta}{2}e^{\lambda(T-(\tilde{t})}(\tilde{x}^2+\tilde{y}^2)-A(t,\tilde{x},u,p_x,\tilde{X},\phi),u(\tilde{t},\tilde{x})-\tilde{x}\Big)\geq 0 
\end{align}
and
\begin{align}\label{eq:two}
\min\Big(&av(\tilde{t},\tilde{y})-p_0 -A((\tilde{t},\tilde{y},p_y,\tilde{Y},\phi) ,v(\tilde{t},\tilde{y})-\tilde{y}\Big)\leq 0 
\end{align}
where 
\begin{align*}
A(\tilde{t},\tilde{x},p,X,\phi)&:=\sup_{q\leq 1} \Big \{L(\tilde{t},\tilde{x},p,X)+K_{\delta}[u,\phi,\tilde{x}]+K^{\delta}[u,\phi,\tilde{x}]\Big \}
\end{align*}
and
\begin{align*}
K_{\delta}[u,\phi,x]&:=\int_{|z|\leq \delta}\phi(t,x+\zeta(x,z))-\phi(t,x)-\zeta(x,z)\partial_x \phi(t,x) \tilde \nu(dz)\\
K^{\delta}[u,\phi,x]&:=\int_{|z|>\delta}u(t,x+\zeta(x,z))-u(t,x)-\zeta(x,z) p_x \tilde \nu(dz)\\
\zeta(x,z)&=(q-x)\frac{e^z-1}{e^z}
\end{align*}
and
\begin{align*}
L(t,x,p,X):=\sigma^2\frac{(q-x)^2}{2}X-(q-x)p
\end{align*}
Subtracting the two inequalities \eqref{eq:one} and \eqref{eq:two}, 
and remarking that $min(a, b)  - min(d, e) \leq 0$ implies either $(a-d)  \leq 0$ or
$(b-e) \leq 0$, we divide our consideration into two cases:

\begin{itemize}
\item \underline{The first case when}:
 \end{itemize}
\begin{align*}
a(u(t,\tilde{x})-v(t,\tilde{y}))+\lambda \frac{\theta}{2}e^{\lambda(T-t)}(\tilde{x}^2+\tilde{y}^2)&\leq A(t,\tilde{x},p_x,\tilde{X})-A(t,\tilde{y},p_y,\tilde{Y})
\end{align*}
We have the first estimates:
\begin{align}\label{L}
L(t,\tilde{x},p_x)-L(t,\tilde{y},p_y)\leq C\Big[\frac{|\tilde{x}-\tilde{y}|^2}{\alpha}+\frac{\theta}{2}e^{\lambda(T-t)}(1+\tilde{x}^2+\tilde{y}^2) \Big]
\end{align}
Since $\phi$ is a $\mathcal{C}^2$-function we have:
\begin{align*}
\phi(t,x+\zeta(x,z))-\phi(t,x)-\zeta(x,z)\partial_x\phi(t,x)\leq C|q-x|^2
\end{align*}
then 
\begin{align*}
\limsup_{\delta \to 0}K_{\delta}[v,\phi,\tilde{x}]\leq 0
\end{align*}
From the definition of $\Psi$ \eqref{eq:three} we have:
\begin{align*}
I&=u(\tilde{t},\tilde{x}+\zeta(\tilde{x},z))-u(\tilde{t},\tilde{x})-\zeta(\tilde{x},z)\big(\frac{|\tilde{x}-\tilde{y}|}{\alpha}+\theta e^{\lambda(T-\tilde{t})}\tilde{x}) \big)\tilde \nu(dz)\\
&-v(\tilde{t},\tilde{y}+\zeta(\tilde{y},z))-v(\tilde{t},\tilde{y})+\zeta(\tilde{y},z)\big(\frac{|\tilde{x}-\tilde{y}|}{\alpha}-\theta e^{\lambda(T-\tilde{t})}\tilde{y})\big) \tilde \nu(dz)\\
&=\Psi(\tilde{t},\zeta(\tilde{x},z),\zeta(\tilde{y},z))-\Psi(\tilde{t},\tilde{x},\tilde{y})\\
&+\frac{1}{\alpha}|\frac{e^z-1}{e^z}|^2|\tilde{x}-\tilde{y}|^2+\theta e^{\lambda(T-\tilde{t})}|\frac{e^z-1}{e^z}|^2(|q-\tilde{x}|^2+|q-\tilde{y}|^2).
\end{align*}
Since $(\tilde{t}, \tilde{x}, \tilde{y})$ is a maximum point of $\Psi$ then $\Psi(\tilde{t},\tilde{x}+\zeta(\tilde{x},z),\tilde{y}+\zeta(\tilde{y},z))-\Psi(\tilde{t},\tilde{x},\tilde{y})\leq 0$, therefore
\begin{align}\label{K}
K^{\delta}\leq C\Big( \frac{1}{\alpha}|\tilde{x}-\tilde{y}|^2+\frac{\theta}{2}e^{\lambda(T-\tilde{t})}(1+|\tilde{x}|^2+|\tilde{y}|^2)\Big).
\end{align}
Therefore we can conclude that 
\begin{align}\label{eq:3}
a(u(\tilde{t},\tilde{x})-v(\tilde{t},\tilde{y}))+\lambda \theta e^{\lambda(T-\tilde{t})}(\tilde{x}^2+\tilde{y}^2)\leq C\Big[\frac{|\tilde{x}-\tilde{y}|^2}{\alpha}+\frac{\theta}{2}e^{\lambda(T-\tilde{t})}(1+\tilde{x}^2+\tilde{y}^2) \Big]
\end{align}
knowing that $\Psi(t,x,x)\leq \Psi(\tilde{t},\tilde{x},\tilde{y})$ then  
\begin{align}\label{eq:four}
u(t,x)-v(t,x)-\theta e^{\lambda(T-t)}x^2\leq u(\tilde{t},\tilde{x})-v(\tilde{t},\tilde{y})
\end{align}
Using \eqref{eq:3} we have
\begin{align*}
u(t,x)-v(t,x)-\frac{\theta}{2} e^{\lambda(T-t)}x^2&\leq 2C\Big[\frac{|\tilde{x}-\tilde{y}|^2}{a\alpha}+\frac{\theta}{2a}e^{\lambda(T-t)}(1+\tilde{x}^2+\tilde{y}^2) \Big]\\
&-\frac{\lambda}{2a}\theta e^{\lambda(T-t)}(\tilde{x}^2+\tilde{y}^2)
\end{align*}
Sending $\alpha \to 0$ we had $(\tilde{t},\tilde{x},\tilde{y})\to (t_0,x_0,x_0)$ and
\begin{align*}
u(t,x)-v(t,x)-\frac{\theta}{2} e^{\lambda(T-t)}x^2&\leq C\frac{\theta}{2a}e^{\lambda(T-t_0)}(1+2x_0^2) \\
&-\frac{\lambda}{a}\theta e^{\lambda(T-t_0)}x_0^2
\end{align*}
By taking $\lambda$ big enough we have
\begin{align*}
u(t,x)-v(t,x)-\theta e^{\lambda(T-t)}x^2&\leq 0
\end{align*}
By sending $\theta$ to zero we have
\begin{align*}
u(t,x)\leq v(t,x)
\end{align*}
\begin{itemize}
\item \underline{The second case}: we have
 
\begin{align*}
u(\tilde{t},\tilde{x})- v(\tilde{t},\tilde{y})+\Big( \tilde{x}-\tilde{y} \Big) \leq 0
\end{align*}
Using \eqref{eq:four} and letting $\alpha$ and $\theta$ go to zero we find out that $u(t,x)\leq v(t,x)$.
\end{itemize}
\end{proof}
\begin{corollary}
$u$ viscosity solution of \eqref{HJB} is unique.
\end{corollary}
\noindent The following proposition discusses the convexity of viscosity solutions using the comparison technique.
\begin{proposition}
Let $u$ $\in C([0, T ] \times\mathbb{R})$ be a viscosity solution of \eqref{HJB}. If $u(T,x) $ is convex for all $x\in \mathbb{R}$ then $u(t, x)$ is convex for all $[0, T ] \times \mathbb{R}.$
\end{proposition}
\begin{proof}
Let $\Psi(t,x,y,z)=u(t,z)-\lambda u(t,x)-(1-\lambda)u(t,y)-\psi(t,x,y,z)$ a function  where $(t,x,y,z)\in[0, T ] \times \mathbb{R}^3$ and
\begin{align}\label{eq:three}
\psi(t,x,y,z)=\frac{|z-\lambda x-(1-\lambda) y|^2}{2\alpha}+\frac{\theta}{2}e^{\rho(T-t)}(z^2+x^2+y^2)
\end{align}
Since $u \in \mathcal{C}([0, T ] \times \mathbb{R}$, $\Psi$ admits a maximum at  $(\tilde{t}, \tilde{x},\tilde{y},\tilde{z})$ such that is a global maximum point of $\Psi$.  As in theorem \eqref{theorem2} we have that along a subsequence  $(t,\tilde{x},\tilde{y},\tilde{z})$ converges to $(t_0, x_0, y_0,z_0) \in [0, T ] \times \mathbb{R}^3$ as $\alpha$ go to zero and $z_0=\lambda x_0+(1-\lambda)y_0$.\\
\noindent If $\bar{t}= T$, we had that $\Psi(t,x,y,\lambda x+(1-\lambda )y) \leq \Psi(T,\tilde{x},\tilde{y},\tilde{z})$ for all $(t,x,y)\in [0,T]\times \mathbb{R}^2$
\begin{align*}
u(t,\lambda x+(1-\lambda )y)-\lambda u(t,x)-(1-\lambda)u(t,y)&-\frac{\theta}{2}e^{\rho(T-t)}(z^2+x^2+y^2)\\
&\leq  u(T,\tilde{z})-\lambda u(T,\tilde{x})-(1-\lambda)u(T,\tilde{y})
\end{align*}
Sending $\alpha$ and $\theta$ gives the convexity of $u$ knowing that $u$ is convex in $T$. In the following we suppose $\tilde{t}\leq T$. Applying Theorem 8.3 of Crandall-Ishii  \cite{Grandall} to the function $\Psi$ at point $ (\tilde{t},\tilde{x},\tilde{y},\tilde{z})$, we can find real numbers $p_0\in \mathbb{R}^3 $ and $(X,  Y, Z) \in \mathbb{R}^3$
such that
\begin{align*}
\Big(p_{0x} ,-p_x,X\Big)\in J^{2,-}\lambda u(\tilde{t},\tilde{x})&\\
\Big(p_{0y},-p_y,Y \Big)\in J^{2,-}(1-\lambda) u(\tilde{t},\tilde{y})&\\
\Big(p_{0z}-\frac{\theta}{2}\rho e^{\rho(T-t)}(\tilde{z}^2+\tilde{x}^2+\tilde{y}^2),p_z,Z\Big)\in J^{2,+}u(\tilde{t},\tilde{z})&
\end{align*}
where
$p_{0x}+p_{0y}-p_{0z}=0$, 
and
\begin{align*}
p_{x_i}:= \partial_{x_i}\psi(t,x_1,x_2,x_3)
\end{align*}
We know that $u$ is a viscosity solution and using the same notation as in theorem \eqref{theorem2} we have:
\begin{align}\label{eq:prop31}
\min\Big(&au(\tilde{t},\tilde{x})-\frac{1}{\lambda}\big(p_{0x}+A(t,\tilde{x},u,p_x,X)\big) ,u(\tilde{t},\tilde{x})-\tilde{x}\Big)\geq 0 
\end{align}
and
\begin{align}\label{eq:prop32}
\min\Big(&au(\tilde{t},\tilde{y})-\frac{1}{1-\lambda}\big(p_{0y}+A(t,\tilde{y},u,p_y,Y)\big) ,u(\tilde{t},\tilde{y})-\tilde{y}\Big)\geq 0 
\end{align}
finally
\begin{align}\label{eq:prop33}
\min\Big(&au(\tilde{t},\tilde{z})+\rho \frac{\theta}{2} e^{\rho(T-t)}(\tilde{x}^2+\tilde{y}^2+\tilde{z}^2)+p_{0,z}+A(t,\tilde{z},u,p_z,Z) ,u(\tilde{t},\tilde{z})-\tilde{z}\Big)\leq 0 
\end{align}
Next, take \eqref{eq:prop33} minus $\lambda$ times \eqref{eq:prop31} minus $(1 - \lambda)$ times \eqref{eq:prop32}, and remarking that $min(a, b)  - \lambda min(d, e)-(1-\lambda)min(f,g) \leq 0$ implies either $(a-\lambda d-(1-\lambda) f)  \leq 0$ or
$(b-\lambda e-(1-\lambda) g) \leq 0$, we divide our consideration into two cases:we divide our consideration into two cases.
 \begin{itemize}
\item \underline{The first case}:
 \end{itemize}
\begin{align*}
a(u(\tilde{t},\tilde{z})-\lambda u(\tilde{t},\tilde{y})-(1-\lambda)u(\tilde{t},\tilde{y}))&+\rho \theta e^{\rho(T-\tilde{t})}(\tilde{x}^2+\tilde{y}^2+\tilde{z}^2)\\
&\leq A(\tilde{t},\tilde{z},p_z,Z)-A(\tilde{t},\tilde{x},p_x,X)- A(\tilde{t},\tilde{y},p_y,Y)
\end{align*}
then using \eqref{L} and \eqref{K} we have
\begin{align}\label{eq88}
a(u(\tilde{t},\tilde{z})-\lambda u(\tilde{t},\tilde{y})-(1-\lambda)u(\tilde{t},\tilde{y}))&+\rho \frac{\theta}{2} e^{\rho(T-\tilde{t})}(\tilde{x}^2+\tilde{y}^2+\tilde{z}^2)\nonumber \\
&\leq C\Big(\frac{|\tilde{z}-\lambda \tilde{x}-(1-\lambda) \tilde{y}|^2}{2\alpha}+\frac{\theta}{2}e^{\rho(T-\tilde{t})}(\tilde{z}^2+\tilde{x}^2+\tilde{y}^2)\Big)
\end{align}
We know that $\Psi(t,x,y,\lambda x+(1-\lambda )y) \leq \Psi(\tilde{t},\tilde{x},\tilde{y},\tilde{z})$ for all $(t,x,y)\in [0,T]\times \mathbb{R}^2$
\begin{align}\label{eq22}
u(t,\lambda x+(1-\lambda )y)-\lambda u(t,x)-(1-\lambda)u(t,y)&-\frac{\theta}{2}e^{\rho(T-t)}(z^2+x^2+y^2)\nonumber \\
&\leq  u(\tilde{t},\tilde{z})-\lambda u(\tilde{t},\tilde{x})-(1-\lambda)u(\tilde{t},\tilde{y})
\end{align}
Taking in account the inequality \eqref{eq88} in \eqref{eq22} and taking  $\alpha \to 0$ we had $(\tilde{t},\tilde{x},\tilde{y},\tilde{z})\to (t_0,x_0,y_0,\lambda x_0+(1-\lambda)y_0)$ and
\begin{align*}
u(t,\lambda x-(1-\lambda) y)-\lambda u(t,x)-(1-\lambda)u(t,y)-\frac{\theta}{2}e^{\rho(T-t)}(z^2+x^2+y^2)&\leq C\frac{\theta}{a}e^{\rho(T-t_0)}(1+x_0^2+y_0^2+z_0^2) \\
&-\frac{\rho}{2a}\theta e^{\rho(T-t_0)}(x_0^2+y_0^2+z_0^2)
\end{align*}
By taking $\rho$ big enough we have
\begin{align*}
u(t,\lambda x-(1-\lambda) y)-\lambda u(t,x)-(1-\lambda)u(t,y)-\frac{\theta}{2}e^{\rho(T-t)}(z^2+x^2+y^2)&\leq 0
\end{align*}
By sending $\theta$ to zero we conclude
\begin{align*}
u(t,\lambda x-(1-\lambda) y)\leq \lambda u(t,x)+(1-\lambda)u(t,y)
\end{align*}
 \begin{itemize}
\item \underline{The second case}:

\begin{align*}
u(\tilde{t},\tilde{z})- \lambda u(\tilde{t},\tilde{x})-(1-\lambda)u(\tilde{t},\tilde{y})-\Big(\tilde{z}-\lambda \tilde{x} -(1-\lambda) \tilde{y}\Big) \leq 0
\end{align*}
Using \eqref{eq22} and letting $\alpha$ and $\theta$ go to zero we find that $u$ is convex
 \end{itemize}
\end{proof}
\bibliographystyle{elsarticle-num}


\end{document}